\documentclass[final,12pt]{colt2022} 

\usepackage{enumerate}
\usepackage{float}
\usepackage{bbm}
\usepackage{caption}
\usepackage[utf8]{inputenc}
\usepackage[english]{babel}
\usepackage{dsfont}	
\usepackage{tikz}
\usepackage{hyperref}
\usepackage{dirtytalk}
\usetikzlibrary{arrows,arrows.meta}

\def\m{\mathcal}

\def\dtv{d_{\mathsf{TV}}}
\def\Bern{\mathsf{Bern}}
\def\unif{\mathsf{Unif}}

\def\m{\mathcal}

\DeclareMathOperator{\Prob}{\mathbf{P}}
\DeclareMathOperator{\mix}{\mathsf{T_{hit}}}
\DeclareMathOperator{\E}{\mathbb{E}}
\DeclareMathOperator{\Pe}{\mathsf{P_e}}

\DeclareMathOperator{\col}{\text{C}_{oll}}

\DeclareMathOperator{\isit}{\mathsf{RUNS}}

\allowdisplaybreaks
\newcommand{\veps}{\varepsilon}

\title[On The Memory Complexity of Uniformity Testing]{On The Memory Complexity of Uniformity Testing}
\usepackage{times}

\coltauthor{%
 \Name{Tomer Berg} \Email{tomerberg@mail.tau.ac.il}\\
 \addr School of Electrical Engineering, Tel Aviv University
 \AND
 \Name{Or Ordentlich} \Email{or.ordentlich@mail.huji.ac.il}\\
 \addr School of Computer Science and Engineering, Hebrew University of Jerusalem
 \AND
 \Name{Ofer Shayevitz} \Email{ofersha@eng.tau.ac.il}\\
 \addr School of Electrical Engineering, Tel Aviv University}
\begin{document}

\maketitle

\begin{abstract}%
   In this paper we consider the problem of uniformity testing with limited memory. We observe a sequence of independent identically distributed random variables drawn from a distribution $p$ over $[n]$, which is either uniform or is $\veps$-far from uniform under the total variation distance, and our goal is to determine the correct  hypothesis. At each time point we are allowed to update the state of a finite-memory machine with $S$ states, where each state of the machine is assigned one of the hypotheses, and we are interested in obtaining an asymptotic probability of error at most $0<\delta<1/2$ uniformly under both hypotheses. 
   The main contribution of this paper is deriving upper and lower bounds on the number of states $S$ needed in order to achieve a constant error probability $\delta$, as a function of $n$ and $\veps$, where our upper bound is $O(\frac{n\log n}{\veps})$ and our lower bound is $\Omega (n+\frac{1}{\veps})$. Prior works in the field have almost exclusively used  collision counting for upper bounds, and the Paninski mixture for lower bounds. Somewhat surprisingly, in the limited memory with unlimited samples setup, the optimal solution does not involve counting collisions, and the Paninski prior is not hard. Thus, different proof techniques are needed in order to attain our bounds.           
\end{abstract}


\section{Introduction}
During the past couple of decades, distribution testing has developed into a mature field. One of the most fundamental tasks in this field is the problem of \textit{uniformity testing}. In this problem, the objective is to design a sample-optimal algorithm that, given i.i.d samples from a discrete distribution $p$ over $[n]$ and some  parameter $\veps>0$, can distinguish (with probability $0<\delta<1/2$) between the case where $p$ is uniform and the case where $p$ is $\veps$-far from uniform in total variation. The common benchmark for characterizing the statistical hardness of uniformity testing is the \textit{minimax sample complexity}, namely the number of independent samples from the distribution one needs to see in order to guarantee that the expected $0-1$ loss is at most $\delta$. However, as the amount of available data is constantly increasing, collecting enough samples for accurate estimation is becoming less of a problem, and it is the dedicated computational resources that constitute the main bottleneck on the performance that a testing algorithm can attain. 

As a result, the topic of testing under computational constraints is currently drawing considerable attention, and in particular, the problem of testing under memory constraints has been recently studied in various different setups, as we elaborate in Section~\ref{subsec:related_work}. In order to single out the effects limited memory has on our algorithms, we let the number of samples we process be arbitrarily large. Perhaps surprisingly, the bounded memory problem does not become trivial in the unbounded sample regime. Moreover, the optimal solution to the problem without memory constraints does not translate to the memory constrained case, as we shall later see. 

Let us formally define the problem of interest: $X_1,X_2,\ldots$ is a sequence of independent identically distributed random variables drawn from a distribution $p$ over $[n]$, which is either uniform or is $\veps$-far from uniform under the total variation distance. Let $\m{H}_0$ denote the hypothesis $p= u$, and $\m{H}_1$ the composite hypothesis $p \in \Theta$, where $\Theta \triangleq \{q:\dtv(q,u) > \veps\}$, where $u$ is the uniform distribution over $[n]$. An \textit{$S$-state uniformity tester} consists of two functions: $f$, and $d$, where $f:[S] \times [n] \rightarrow [S]$ is a state transition (or memory update) function, and
$d:[S]\rightarrow \{\m{H}_0,\m{H}_1\}$ is a decision function. Letting $M_t$ denote the state of the memory at time $t$, this finite-state machine evolves according to the rule 
\begin{align}
M_0&=s_{\text{init}},\label{eq:init} \\M_t&=f(M_{t-1},X_t)\in [S],\label{eq:evolution}
\end{align}
for some predetermined initial state $s_{\text{init}} \in [S]$. If the machine is stopped at time $t$, it outputs the decision $d(M_t)$. Note that here the memory update function $f$ is not allowed to depend on time. The restriction to time-invariant algorithms is operationally appealing, since storing the time index necessarily incurs a memory cost. Furthermore, since the number of samples is unbounded, simply storing the code generating a time-varying algorithm may require unbounded memory. We define the asymptotic expected $0-1$ loss under each hypothesis (also known as the type-$1$ and type-$2$ probabilities of error) as
	\begin{align}
	    \Pe (f,d|\m{H}_0) &= \lim_{t\rightarrow\infty} \frac{1}{t}\sum_{i=1}^t \Pr (d(M_i)=1|X_1,X_2,\ldots \overset{iid}{\sim}  u), \\
	    \Pe (f,d|\m{H}_1) &= \underset{p \in \Theta}{\sup}\lim_{t\rightarrow\infty} \frac{1}{t}\sum_{i=1}^t \Pr (d(M_i)=0|X_1,X_2,\ldots \overset{iid}{\sim} p)\label{eq:pefd}. 
	\end{align}
We are interested in the \emph{minimax memory complexity} of the tester, $S^*(n,\veps)$, defined as the smallest integer $s$ for which there exist $(f,d)$ such that $\Pe (f,d|\m{H}_i)\leq \delta$, for any $i\in \{0,1\}$. In previous works dealing with memory limited testing/estimation problems, the memory complexity was taken to be either the number of bits or the number of memory words needed to solve the problem, up to multiplicative constants. In those cases, a large gap between, for example, $S=n$ and $S=n^2$ states, would correspond only to a factor of $2$ in the number of bits we need to store. These definitions of memory complexity are natural in the large memory regime. However, in this work, we consider the small memory regime, so the definition in terms of the number of states seems more appropriate to us. In such regime, one might not be able to increase the number of available memory bits even by a constant factor. In other words, we are interested in deriving additive constant bounds on the number of bits stored, rather than constant factor bounds.


The sample complexity of uniformity testing is well known, as a result of the works of~\cite{goldreich2000testing} and~\cite{paninski2008coincidence}, and is $\Theta(\sqrt{n}/\veps ^2)$. The lower bound is established using $\chi^2$ divergences, where the input is drawn from the Paninski mixture. The upper bound is based on collision estimators, since $\sqrt{n}$ samples will result in $1$ expected collision on average under the uniform distribution, and $1+\Omega(\veps^2)$ expected collisions on average under the alternative, and the variances are small enough to distinguish between the hypotheses. Collision estimators have henceforth become the default algorithmic tool whenever uniformity testers are involved, and most upper bounds in the field rely on some variant of the collision estimator. In a similar fashion, the Paninski prior has become a recurring ingredient in lower bounds for the problem, since it is difficult to identify with limited samples. As it turns out, when memory is the more stringent restriction, these approaches are lacking. This can be observed in a recent work of~\cite{diakonikolas2019communication} that is concerned with the sample complexity under memory constraints of uniformity testing, and in which the memory $S>>n^{1/\veps^2}$ considered by the authors is substantially greater than what our results imply, even though the authors work with the more lenient streaming algorithms, that allow for time-varying functions. The increased required memory in~\cite{diakonikolas2019communication} mainly stems from the fact that the algorithm considered there is based on collision estimation, which is more wasteful in general. On the other hand, the sample complexity attained by that algorithm is much smaller than the number of samples the algorithm proposed here requires, but this reduced sample size comes at the price of increased memory (we discuss this in Section~\ref{sec:upper}).

Our main results are twofold. The first is showing that reduction of the uniformity testing problem to sequential binary hypothesis testing results in a substantial saving in memory, with respect to collision estimation. The performance of a scheme based on this idea is analyzed, and gives rise to an upper bound on the memory complexity (Theorem~\ref{thm:upper_bound}).  The second result consists of two lower bounds on the memory complexity. For the first lower bound, we use a member of the Paninski prior family to show that $S^*(n,\veps)=\Omega(1/\veps)$. However, even for the full Paninski mixture, we show that  $O(1/\veps)$ states are sufficient for reliably testing against the uniform distribution, when the number of available samples is large. Consequently, we must choose a different approach for our $\Omega(n)$ lower bound. Through analysis of stationary distributions of Markov processes, we show that for any uniformity tester with insufficient memory there is some distribution $q$ with $\dtv(q,u)>\veps$ which induces the same stationary distribution on its memory states as the uniform distribution, hence is indistinguishable from it (Theorem~\ref{thm:lower_bound}). This is in contrast to the standard approach of finding a suitable prior on $\Theta$ under which the Bayesian problem is hard.

\begin{theorem}\label{thm:upper_bound}
	\begin{align}
    S^*(n,\veps) = O \left(\frac{n \log n}{\veps}\right).
	\end{align}
\end{theorem}
It is worth noting that the algorithm achieving this upper bound is deterministic. Our lower bound holds for randomized algorithms.
\begin{theorem}\label{thm:lower_bound}
	\begin{align}
	S^*(n,\veps) =  \Omega \left(n+\frac{1}{\veps} \right).   
	\end{align}
\end{theorem}
We note that both bounds hold also for the hypothesis testing problem of the collision probability, in which we test between collision probability that is either $\frac{1}{n}$ or greater than $\frac{1+\veps ^2}{n}$.

\section{Related work}\label{subsec:related_work}\label{sec:related}

The study of learning and estimation under memory constraints has been initiated in the late 1960s by Cover and Hellman (with a precursor by~\cite{robbins1956sequential}) and remained an active research area for a decade or so. Most of the old work on learning with finite memory has been focused on the hypothesis testing problem. For the problem of deciding whether an i.i.d. sequence was drawn from $\Bern(p)$ or $\Bern(q)$,~\cite{cover1969hypothesis} described a time-varying finite-state machine with only $S=4$ states, whose error probability approaches zero with the sequence length. As time-varying procedures suffer from the shortcomings described earlier,~\cite{hellman1970learning} addressed the same binary hypothesis testing problem within the class of \textit{time-invariant randomized} procedures. They have found an \emph{exact} characterization of the smallest attainable error probability for this problem. In a recent paper,~\cite{berg2020binary} derived a lower bound on the error probability attained by any $S$-state deterministic procedure, showing that while the smallest attainable error probability decreases exponentially fast with $S$ in both the randomized and the deterministic setups, the base of the exponent can be arbitrarily larger in the randomized case.

The work on the uniformity testing problem has been initiated by~\cite{goldreich2000testing}, who proposed a simple and natural uniformity tester that relies on the \textit{collision probability} of the unknown distribution, which is the probability that two samples drawn according to $p$ are equal, and succeeds after drawing $O(\sqrt{n}/\veps^4)$ samples. In subsequent work,~\cite{paninski2008coincidence} showed an information-theoretic lower bound of $\Omega(\sqrt{n}/\veps^2)$ on the sample complexity and provided a matching upper bound that holds under some assumption on $\veps$, which has been later shown to be unnecessary~\cite{diakonikolas2014testing}. It was recently shown in~\cite{goldreich2016uniform} that the more general problem of distribution testing can be reduced to uniformity testing with only a constant factor loss in sample complexity. In~\cite{acharya2019inference}, the authors addressed the distributed variant of the problem: each player receives one sample from the distribution, about which they can only provide limited information to a central referee. 

The space-complexity (which is the minimal memory in bits needed for the algorithm) of estimating the empirical collision probability is well-studied for worst case data streams of
a given length, dating back to the seminal work of~\cite{alon1999space}. In~\cite{crouch2016stochastic} the authors studied the trade-off between sample complexity and space complexity for the problem of estimating the collision probability. 
The trade-off between sample complexity and space/communication complexity for the distribution testing problem has recently been addressed by~\cite{diakonikolas2019communication}, who gave upper and lower bounds on the sample complexity under memory constraints. On a broader level, the problem of estimating statistics with bounded memory is receiving more and more attention in the machine learning literature lately, see, e.g.,~\cite{chien2010space,mcgregor2012space,kontorovich2012statistical,sd15,svw16,raz18,ds18,dks19,ssv19}. Another closely related active line of work is that of estimating statistics under limited communication, e.g.,~\cite{zdjw13,garg2014communication,bgmnw16,xr17,jordan2018communication,han2018geometric,how18,bho18,act18,hadar2019communication,hadar2019distributed,acharya2020inference}.
\section{Preliminaries}
In this section, we introduce the mathematical notation and background necessary to state
and prove our results in the following sections.
\subsection{Notations}
We write $[n]$ to denote the set $\{1,\ldots,n\}$, and consider discrete distributions over $[n]$. We use the notation $p_i$ to denote the probability of element $i$ in distribution $p$, and the notation $p_i^{(i,j)}$ to denote the probability of element $i$ restricted to $(i,j)$, that is, $p_i^{(i,j)}=\frac{p_i}{p_i+p_j}$. The $\ell^k$ norm of a distribution is $||p||_k=\sqrt[k]{\sum_{i=1}^n |p_i|^k}$. The total variation distance between distributions $p$ and $q$ is defined as half their $\ell^1$ distance, i.e., $\dtv(p,q)=\frac{1}{2}||p-q||_1=\frac{1}{2}\sum_{i=1}^n |p_i-q_i|$.
The collision probability of a distribution $p$ is defined as the squared $\ell_2$ norm of $p$, that is, $\col(p)=||p||_2^2=\sum_{i=1}^n p_i^2 $.
\subsection{Equivalence to collision estimation}
As we mentioned, the problem of uniformity testing with limited samples and unlimited memory is equivalent to the collision testing problem, due to the following two reasons:
\begin{enumerate}
    \item $\dtv{(p,u)}>\veps$ implies $\col (p)>\frac{1+\veps ^2}{n}$, since, using the Cauchy-Schwarz inequality:
    \begin{align}
       \col (p)&\nonumber =\sum_{i=1}^n p_i^2\nonumber\\&=\frac{1}{n}+\sum_{i=1}^n \left(p_i-\frac{1}{n}\right)^2\nonumber\\&\geq \frac{1}{n}+\frac{1}{n}\left(\sum_{i=1}^n\left|p_i-\frac{1}{n}\right|\right)^2\\&=\frac{1+\dtv{(p,u)}^2}{n},\label{eq:coll_lower}
    \end{align}
and we can solve the latter problem with $O(\sqrt{n} / \veps^2)$ samples, hence we can also solve the former with the same number of samples.  
    \item  The lower bound on the sample complexity is proven for the Paninski prior, $q^{\text{Pan}}(z)$, whose value at each coordinate $i\in [n]$ is determined by the vector $z\in \{-1,1\}^{n/2}$, and is given by 
    \begin{align}
        q^{\text{Pan}}_i(z)= \begin{cases}
        \frac{1+\veps z_{i/2}}{n},& i \text{ even}\\\frac{1-\veps z_{(i+1)/2}}{n},& i \text{ odd}. \end{cases}\label{eq:panin}
    \end{align}
    It is not difficult to see that $\dtv{(u,q^{\text{Pan}})}=\veps$ and $\col (q^{\text{Pan}}(z))=\frac{1+\veps ^2}{n}$, for all $z\in \{-1,1\}^{n/2}$. Now taking $z_1,\ldots z_{n/2}$ to be Rademacher random variables, we obtain a distribution on $\Theta$, for which the sample complexity in the Bayesian setting is $\Omega(\sqrt{n} / \veps^2)$. This implies the same lower bound for the composite hypothesis testing problem.
\end{enumerate}

\section{Upper Bound}\label{sec:upper}
Following the above equivalence, it seem only natural to try and obtain upper bounds by utilizing collision-based estimators in our bounded memory framework as well. Naively implementing the collision counting decision rule of~\cite{goldreich2000testing,paninski2008coincidence} requires $\Omega(n^{\sqrt{n}/\veps^2})$ states. However, since the number of samples is unbounded, this algorithm is extremely wasteful. 
As a warm-up, we first propose the following variation, which uses the unlimited number of samples to save memory, by only counting consecutive collisions. Note that the probability that two consecutive samples are equal under $p$ is exactly $\sum_{i=1}^n p_i^2=\col (p)$. We define a \textit{collision machine} to be a machine that, at each time point, outputs a $1$ if the current sample is the same as the preceding sample and $0$ otherwise. This can be implemented using $O(n)$ states, since we are storing an element of the domain (the preceding sample) to check for collisions, and results in a Bernoulli process with parameter $p=\frac{1}{n}$ under the uniform distribution, or a Bernoulli process with parameter $p> \frac{1+\veps ^2}{n}$ under the alternative (due to~\eqref{eq:coll_lower}). Thus, if we can successfully distinguish between two i.i.d Bernoulli processes with parameters that are roughly $\veps^2/n$ apart, our uniformity tester will separate the hypotheses successfully. To that end, we appeal to~\cite{berg2021deterministic}, where we defined the following machine which, given $X_1,X_2,\ldots \sim \Bern(\theta)$, can distinguish between the composite hypotheses $\theta < q$ and $\theta > p$, for $p>q$. \footnote{Their machine was designed to solve the \emph{simple} binary hypothesis test $\mathcal{H}_0:\{\theta=p\}$ vs. $\mathcal{H}_1:\{\theta=q\}$, but the difference between the two problems is not significant.}

\begin{definition}
  $\isit(N,p,q)$ is the machine with $N\geq 4$ states depicted in Figure~\ref{fig:BHT_Machine},
  designed to decide between the hypotheses $\mathcal{H}_0:\{\theta>p\}$ vs. $\mathcal{H}_1:\{\theta<q\}$, for some $0<q<p<1$. The machine is initialized at state $s$ and evolves according to the sequence of input bits $X_1,X_2,\ldots$. If the machine observes a run of $N-s$ ones before observing a run of $s-1$ zeros, it decides $\mathcal{H}_0$, otherwise it decides $\mathcal{H}_1$.
  The initial state of the machine is $s=f(N,p,q)$, where
  \begin{align}
    f(N,p,q) \triangleq  2+\left\lceil \frac{\log pq}{\log p(1-p)+\log q(1-q)}(N-3)\right\rfloor,\label{eq:start_state}
  \end{align}
  is an integer between $2$ and $N-1$. We denote the (worst case) error probability of the machine by $\Pe^{\isit(N,p,q)}=\max\left\{p^0_1,p^1_0\right\}$, where
  \begin{align}
  p^0_1&=\sup_{\theta<q}~~\Pr_{X_1,X_2\ldots\stackrel{i.i.d.}{\sim}\Bern(\theta)}\left(\isit(N,p,q)\text{ decides } \mathcal{H}_0 \right),\\
  p^1_0&=\sup_{\theta>p}~~\Pr_{X_1,X_2\ldots\stackrel{i.i.d.}{\sim}\Bern(\theta)}\left(\isit(N,p,q)\text{ decides } \mathcal{H}_1 \right).
  \end{align}

\end{definition}
\begin{center}
\begin{figure}[H]
\centering
\setlength\belowcaptionskip{-1.4\baselineskip}
\begin{tikzpicture}
  \tikzset{
    >=stealth',
    node distance=0.92cm,
    state/.style={font=\scriptsize,circle, align=center,draw,minimum size=20pt},
    dots/.style={state,draw=none}, edge/.style={->},
  }
  \node [state ,label=center:$1$] (S0)  {} ;
  \node [state] (S0-1)   [right of = S0]   {};
  \node [dots] (dots1)   [right of = S0-1]   {$\cdots$};
  \node [state] (1l) [right of = dots1]  {};
  \node [state ,label=center:$s$] (0)   [right of = 1l] {};
  \node [state] (1r) [right of = 0]   {};
  \node [dots]  (dots2)  [right of = 1r] {$\cdots$};
  \node [state] (S1-1) [right of = dots2]  {};
  \node [state ,label=center:$N$] (S1) [right of = S1-1]  {};
  \path [->,draw,thin,font=\footnotesize]  (S0-1) edge[bend left=45] node[below ] {$0$} (S0);
  \path [->,draw,thin,font=\footnotesize]  (0) edge[bend left=45] node[below] {$0$} (1l);
  \path [->,draw,thin,font=\footnotesize]  (1r) edge[bend left=45] node[below right] {$0$} (1l);
  \path [->,draw,thin,font=\footnotesize]  (S1-1) edge[bend left=45] node[below right] {$0$} (1l);
  \path [->,draw,thin,font=\footnotesize]  (1l) edge[bend left=45] node[below ] {$0$} (dots1);
  \path [->,draw,thin,font=\footnotesize]  (S0-1) edge[bend left=45] node[above left] {$1$} (1r);
  \path [->,draw,thin,font=\footnotesize]  (1l) edge[bend left=45] node[above left] {$1$} (1r);
  \path [->,draw,thin,font=\footnotesize]  (0) edge[bend left=45] node[above ] {$1$} (1r);
  \path [->,draw,thin,font=\footnotesize]  (1r) edge[bend left=45] node[above] {$1$} (dots2);
   \path [->,draw,thin,font=\footnotesize]  (S0) edge[loop left] node[below] {$0/1$} (S0);
  \path [->,draw,thin,font=\footnotesize]  (S1-1) edge[bend left=45] node[above ] {$1$} (S1);
 \path [->,draw,thin,font=\footnotesize]  (S1)  edge[loop right]  node[above ]{$0/1$} (S1);

\end{tikzpicture}
\caption{$\isit(N,p,q)$ - Deterministic Binary Hypothesis Testing Machine  }\label{fig:BHT_Machine}
\end{figure}
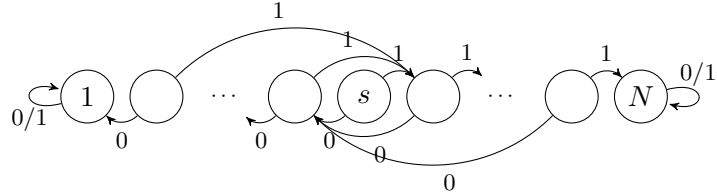
\end{center}

In the same paper, we have shown that with $N=O(K)$ states, the machine $\isit(N,p,q)$ can decide whether $\theta>p$ or $\theta<q=p-1/K$ with constant error probability $\delta<1/2$.

\begin{lemma} [~\cite{berg2021deterministic}]\label{lem:N_eps}
For any $\frac{2}{K}\leq  p\leq 1-\frac{1}{K}$, $q=p-\frac{1}{K}$ and $0<\delta<1/2$, let
\begin{align}
N = N(\delta,p,K) \triangleq 3+\left\lceil K\cdot 6\log\frac{2}{\delta\cdot\left(p-\frac{1}{K}\right)(1-p)} \right\rceil.
\label{eq:NepK}
\end{align}
Then
\begin{align}
\Pe^{\isit(N,p,q)}<\delta.
\end{align}
\end{lemma}       
We can thus separate the two hypotheses with $O(n/\veps^2)$ states, so overall we need the product of the two parts, which is $O(n^2/\veps^2)$ states, a memory size that is \textit{quadratic} in $n/\veps$. We note that in order to distinguish hypotheses that are $1/K$ apart, we must have $\Omega(K)$ states [~\cite{hellman1970learning}, Theorem $3$], therefore once the reduction to pairwise collisions has been done, using $\isit$ is essentially the best possible, up to constants.
It is worth mentioning that the effective sample complexity of this machine (the expected number of samples needed for the algorithm to converge) is $O(n/\veps^2)$, which is indeed greater then the $O(\sqrt{n}/\veps^2)$ sample complexity without memory constraints. This follows since $n$ samples will result in $1$ collision on average under the uniform distribution, and $1+\Omega(\veps^2)$ collisions on average under the alternative. However, the memory used is substantially smaller than  $n^{1/\veps^2}$, the smallest considered by~\cite{diakonikolas2019communication}. As it turns out, we can find good uniformity testers with even smaller memory.   
In the following, we describe a uniformity tester with a memory size that is only \textit{linear} in $n/\veps$, up to logarithmic factors.

\subsection{Proof of Theorem~\ref{thm:upper_bound}}
In a nutshell, our algorithm reduces uniformity testing (and collision estimation) to a sequence of binary tests. The intuition here is that if $\dtv(p,u) > \veps$ then there must be some symbols $(i,j) \in [n] \times [n]$ such that either $\Pr(X=i | X \in \{i,j\})$ or $\Pr(X=j | X \in \{i,j\})$ are sufficiently biased, where by sufficiently we mean with some bias $\Omega(\veps)$. 
We could then clearly run binary tests over all $n^2$ pairs of symbols, but it turns out we can reduce this significantly to just $2n$ tests. One way to do this is to test, for all $i>1$, whether $\Pr(X=1 | X \in \{1,i\})$ or $\Pr(X=i | X \in \{1,i\})$ are sufficiently biased. 
If $p=u$, none of these probabilities are sufficiently biased, whereas if $\dtv(p,u) > \veps$, at least one of them is (we prove this in the sequel). Of course, this machine can get stuck for an arbitrarily long time if $\Pr(X \in \{1,i\})$ is extremely small (or even zero) for some $i>1$, as it will be waiting for an arbitrarily low probability event. To circumvent this issue, we ensure that the algorithm terminates at some point by first checking that $\Pr(X=1)$ is not too small. Specifically, we check that it is not much smaller than $n^{-1}$. The reason is since $p$ can only be uniform, in which case all probabilities are $n^{-1}$, or $\veps$-far from uniform. Thus, if $\Pr(X=1)$ is much smaller than $n^{-1}$, we can confidently assert that $p$ is not uniform. Otherwise, we continue to the previously mentioned testing of pairs.  

Let us now describe the particular structure of the proposed machine. In our construction, the state space $[S]$ is partitioned into $K=2n+1$ disjoint sets denoted by $\mathcal{S}_1,\ldots,\mathcal{S}_K$, which we will refer to as \emph{mini-chains}. Each mini-chain performs a binary hypothesis test, and we move sequentially from one mini-chain to the next, until we become confident of our decision, in which case we enter an absorbing state. $\mathcal{S}_1$ is a $\isit(N,p,q)$ machine for testing whether the probability of the symbol $1$, namely $p_1$, is sufficiently large, and its parameters are $q=\frac{1}{4n}$, $p=\frac{1}{2n}$, and
\begin{align}
N_1 = N\left(\frac{\delta}{2},\frac{1}{2n},\frac{1}{4n}\right) \triangleq 3+\left\lceil 6\cdot 4n \log\frac{4\cdot 4n}{\delta\cdot\left(1-\frac{1}{2n}\right)} \right\rceil
\end{align}
states. If it decides in favor of $q$ it outputs $\m{H}_1$ and terminates, whereas if it decides in favor of $p$ it continues to $\m{S}_2$. Mini-chains $\m{S}_2$ to $\m{S}_{2n}$ test consecutively for the conditional distributions $p_1^{(1,2)},p_2^{(1,2)}$ up to $p_1^{(1,n)},p_n^{(1,n)}$ between $q=\frac{1}{2}+\frac{\tilde{\veps}}{2}$ and $p=\frac{1}{2}+\tilde{\veps}$, for $\tilde{\veps}=\frac{\veps}{8+4\veps}$. To that end, we use identical $\isit(N,p,q)$ machines which, when testing for the pair $(1,i)$, change their state only upon observing symbols $1$ or $i$. Let $\Pe\triangleq\frac{\delta/2}{2n}$, and construct the machines to attain worst case error probability of at most $\Pe$. Thus, the number of states needed is, for all $2\leq j \leq 2n+1$,
\begin{align}
N_j = N\left(\Pe,\frac{1}{2}+\tilde{\veps},\frac{\tilde{\veps}}{2}\right) \triangleq 3+\left\lceil \frac{12}{\tilde{\veps}} \log\frac{4}{\Pe\cdot(1+\tilde{\veps})\left(1/2-\tilde{\veps}\right)} \right\rceil
\end{align}
states. Here, if the machine decides in favor of $q$ it moves to the next mini-chain, whereas if it decides in favor of $p$ it outputs $\m{H}_1$ and terminates. The construction of $\m{S}_{2n+1}$ is similar, with the exception that if it decides in favor of $q$ it outputs $\m{H}_0$ and terminates. Thus, the machine only outputs $\m{H}_0$ after all binary tests are completed, and none of the conditional probabilities are found to be too biased. The binary hypothesis testing machines are concatenated left to right, such that the $2i$-th machine, for $i\in [n]$, is as depicted in Figure~\ref{fig:mini_chain}.
\begin{center}
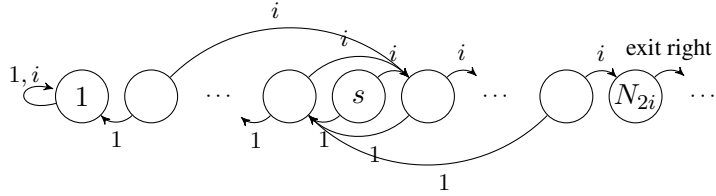
\begin{figure}[H]
\setlength\belowcaptionskip{-1.7\baselineskip}
\centering
\begin{tikzpicture}
  \tikzset{
    >=stealth',
    node distance=0.92cm,
    state/.style={font=\scriptsize,circle, align=center,draw,minimum size=20pt},
    dots/.style={state,draw=none}, edge/.style={->},
  }
  \node [state ,label=center:$1$] (S0)  {} ;
  \node [state] (S0-1)   [right of = S0]   {};
  \node [dots] (dots1)   [right of = S0-1]   {$\cdots$};
  \node [state] (1l) [right of = dots1]  {};
  \node [state ,label=center:$s$] (0)   [right of = 1l] {};
  \node [state] (1r) [right of = 0]   {};
  \node [dots]  (dots2)  [right of = 1r] {$\cdots$};
  \node [state] (S1-1) [right of = dots2]  {};
  \node [state ,label=center:$N_{2i}$] (S1) [right of = S1-1]  {};
   \node [dots]  (dots3)  [right of = S1] {$\cdots$};
  \path [->,draw,thin,font=\footnotesize]  (S0-1) edge[bend left=45] node[below ] {$1$} (S0);
  \path [->,draw,thin,font=\footnotesize]  (0) edge[bend left=45] node[below] {$1$} (1l);
  \path [->,draw,thin,font=\footnotesize]  (1r) edge[bend left=45] node[below right] {$1$} (1l);
  \path [->,draw,thin,font=\footnotesize]  (S1-1) edge[bend left=45] node[below right] {$1$} (1l);
  \path [->,draw,thin,font=\footnotesize]  (1l) edge[bend left=45] node[below ] {$1$} (dots1);
  \path [->,draw,thin,font=\footnotesize]  (S0-1) edge[bend left=45] node[above left] {$i$} (1r);
  \path [->,draw,thin,font=\footnotesize]  (1l) edge[bend left=45] node[above left] {$i$} (1r);
  \path [->,draw,thin,font=\footnotesize]  (0) edge[bend left=45] node[above ] {$i$} (1r);
  \path [->,draw,thin,font=\footnotesize]  (1r) edge[bend left=45] node[above] {$i$} (dots2);
  \path [->,draw,thin,font=\footnotesize]  (S0) edge[loop left] node[above] {$1,i$} (S0);
  \path [->,draw,thin,font=\footnotesize]  (S1-1) edge[bend left=45] node[above ] {$i$} (S1);
  \path [->,draw,thin,font=\footnotesize]  (S1)  edge[bend left=45]  node[above ]{exit right} (dots3);
\end{tikzpicture}
\caption{Illustration of mini-chain $\m{S}_{2i}$  }\label{fig:mini_chain}
\end{figure}
\end{center}
The structure of the $(2i+1)$-th machine, for $1\leq i \leq n-1$, is similar, except that the arrow labels $1$ and $i$ are swapped, since we test for $p_i^{(1,i)}$ instead of $p_1^{(1,i)}$. The only difference in the $2n+1$ machine is the last state, which is an absorbing one. The total number of states is:
\begin{align}
    S=\sum_{j=1}^{2n+1} N_j&=N_1+2n \cdot N_2\\&\leq 4(2n+1)+24n \cdot\log\frac{4\cdot 4n}{\delta\left(1-\frac{1}{2n}\right)}+\frac{24n}{\tilde{\veps}} \cdot\log\frac{4\cdot 4n}{\delta(1+\tilde{\veps})\left(1/2-\tilde{\veps}\right)}\\&\leq 4(2n+1)+24n\left(1+\frac{1}{\tilde{\veps}}\right)\cdot\log\frac{16n}{\delta(1+\tilde{\veps})\left(1/2-\tilde{\veps}\right)}\\&\leq 9n+\frac{10\cdot 24n}{\veps}\cdot\log \frac{40n}{\delta }.
\end{align}
We now show that this is indeed an $(S,\delta)$ uniformity tester. Consider first the case $p=u$, which implies $p_i=\frac{1}{2}$ for all $i$. The probability of outputting $\mathcal{H}_1$ in $\m{S}_1$ is upper bounded by $\frac{\delta}{2}$. The probability of outputting $\mathcal{H}_1$ in any of the consecutive binary tests is the probability we make an error in \textit{at least one} binary test, and is, at most
\begin{align}
    1-(1-\Pe)^{2n} = 1- \left(1- \frac {\delta/2}{2n}\right)^{2n}\leq 1-e^{- \delta/2}\leq \frac{\delta}{2}.
\end{align}
Hence, by the union bound, the total error probability is at most $\delta$.
For the case $p\in \Theta$, we show that if $\dtv{(p,u)}>\veps$, then the probability of outputting $\mathcal{H}_0$ is smaller than $\delta$. To that end, it is sufficient to show that there is some $i$ for which $\max \{p_1^{(i,1)},p_i^{(i,1)}\}>\frac{1}{2}+\tilde{\veps}$, since that would imply the probability that either $\m{S}_{2i}$ or $\m{S}_{2i+1}$ output $\m{H}_1$ is at least $1-\Pe$, which in turn implies that the overall error probability of the algorithm is at most $\Pe=\frac {\delta/2}{2n}<\delta$. Thus, we prove the following lemma in the appendix.
\begin{lemma}\label{lem:biased}
If $\dtv{(p,u)}>\veps$, then there exist some $i \in [n]$ such that
\begin{align}
    \max \{p_1^{(i,1)},p_i^{(i,1)}\}>\frac{1}{2}+\tilde{\veps}.
\end{align}
\end{lemma}
While we leave the investigation of the optimal memory-sample complexity trade off to future work, it is worth mentioning that the effective sample complexity of the above algorithm is $O\left(n^2e^{O(1/\veps)}\right)$. This follows since there are $O(n)$ mini-chains that we traverse in sequence, each is activated every $O(n)$ samples on average, and the mixing time of each one is $e^{O(1/\veps)}$. Thus, if $\veps$ is taken to be a constant and $n$ is the dominating factor, we can loosely say that, with respect to the simpler collision tester proposed earlier, a quadratic saving in memory states is translated to a quadratic increase in the number of expected samples needed for testing.
\section{Lower Bound}
In this section we prove Theorem~\ref{thm:lower_bound}. We prove the $\Omega\left(1/\veps\right)$ bound via reduction to simple binary hypothesis testing. The proof of the $\Omega(n)$ bound is separated to the ergodic and non-ergodic cases, corresponding to irreducible and reducible machines respectively: in the first part we prove the ergodic case using linear algebra tools, and in the second part we show that any good uniformity tester that induces a non-ergodic Markov chain can be reduced to a good uniformity tester that induces an ergodic Markov chain. 
\subsection{Proof of the \texorpdfstring{$\Omega\left(1/\veps\right)$}{Lg} bound}
In this proof, we use the fact that a good uniformity tester is also a good binary hypothesis tester for hypotheses that are $\veps$ apart. Assume we have an ($S,\delta$) uniformity tester. We use this machine to construct a randomized binary hypothesis testing machine that decides whether an independent identically distributed sequence of bits $B_1,B_2,\ldots$ is either $\Bern\left(\frac{1}{2}\right)$ distributed or $\Bern\left(\frac{1-\veps}{2}\right)$ distributed as follows:
\begin{enumerate}
    \item Let $U_i\overset{iid}{\sim} \unif\left(\left[\frac{n}{2}\right]\right)$ and let $B_i\overset{iid}{\sim} \Bern(\theta)$, where $\theta$ is either $\frac{1}{2}$ or $\frac{1-\veps}{2}$.
    \item We generate the random variables $X_i \in [n]$ by setting $X_i=2U_i-B_i$. Note that if $\theta=\frac{1}{2}$, then $X_i \sim u$. On the other hand, if $\theta=\frac{1-\veps}{2}$, then
    \begin{align}
     X_i \sim p = \frac{1}{n}\cdot [1-\veps \hspace{2mm} 1+\veps \hspace{2mm} 1-\veps \hspace{2mm} 1+\veps \hspace{2mm} \ldots], 
    \end{align}
    so that $\dtv(p,u)=\veps$ ($p$ in this case is one of the priors in the Paninski family of~\eqref{eq:panin}).
    \item Thus, our randomized machine can distinguish between $\Bern\left(\frac{1}{2}\right)$ and $\Bern\left(\frac{1-\veps}{2}\right)$. Appealing to~\cite{hellman1970learning}[Theorem $3$], it therefore must have $S=\Omega(1/\veps)$ states. 
\end{enumerate}
\subsection{Proof of the \texorpdfstring{$\Omega(n)$}{Lg} bound for ergodic machines}
As we mentioned earlier, the Paninski prior is not a good choice for lower bounds in the memory limited regime, for the following reason: Any distribution in the Paninski family has either $p_1^{(1,2)}=\frac{1+\veps}{2}$ or $p_2^{(1,2)}=\frac{1+\veps}{2}$. Thus, we can test these two options with two $\isit(N,p=\frac{1+\veps}{2},q=p-\frac{\veps}{2})$ machines with $O(1/\veps)$ states, hence the Paninski mixture cannot give us a lower bound that depends on the dimension $n$. The reason for this is that since the number of samples is unlimited, we can afford to wait for events that sharply distinguish between the hypotheses, and, in particular relevance to the above argument, to wait for appearances of the first two symbols alone. When one cares about sample complexity only, then this approach is obviously very wasteful since many samples are ignored. But when one is concerned with memory complexity only, then waiting for rare binary events turns out to be the right approach, since reducing the dimension of the testing problem is of most importance. 

We first consider the case of ergodic machines, corresponding to irreducible Markov chains. The method of proof is the following: Any finite-state machine induces a Markov chain, that is characterized by some transition matrix $\Prob$.
Our method of proof is based on analysis of stationary distributions and properties of Markov chain transition matrices, specifically their null space. Firstly, we show that for any finite-state machine with $S =  O(\sqrt{n})$ states, there must exist some distribution $p$ over $[n]$ with $\dtv{(p,u)}>\veps$ that induces the same transition matrix as $u$. Secondly, we tighten this result by showing that when $S =  O(n)$ there is some distribution $p$ over $[n]$ with $\dtv{(p,u)}>\veps$ that induces the same stationary distribution as $u$ (but not necessarily the same transition matrix), thus $p$ is indistinguishable from $u$ when viewed through the state sequence of the machine.   
\begin{lemma}\label{lem:sameP}
Denote the transition matrix induced by a distribution $p$ as $\Prob(p)$. Then for any finite-state machine with $S =  O(\sqrt{n})$ there must exist some distribution $p$ with $\dtv{(p,u)}>\veps$ such that $\Prob(u)=\Prob(p)$.
\end{lemma}
\begin{proof}
Fix a finite state machine with $S$ states, and define the matrix $A\in[0,1]^{S^2 \times n}$ whose rows are indexed by $(i,j)\in [S]\times [S]$. The entry $A_{(i,j),m}$ specifies the probability of moving to state $j$ from state $i$ when the input is $m\in[n]$. Note that for deterministic machines the values of $A$ can only be $0$ or $1$, but our lemma in fact holds for the more general case of randomized machines.
Let $p(u)=A u$, where $u=[1/n \ldots 1/n]^T$, be the vectorized form of the transition matrix induced on the machine by the uniform distribution, that is, $p(u)$ is generated by stacking the columns of $\Prob(u)$ on top of each other.
Now, if $S^2<n$, the matrix $A$ must have a nontrivial kernel of rank at least $n-S^2$. Let $V$ be the subspace spanned by $\ker(A)\cap \ker([1 \ldots 1])$. We have that $\dim(V)\geq n-1-S^2 \triangleq k$. Before we continue, we need the following lemma, due to David E. Speyer.\footnote{\url{https://math.stackexchange.com/questions/1850488}}  
\begin{lemma}\label{lem:ptoq}
Let $V$ be a $k$-dimensional linear subspace of $\mathbb{R}^n$. Define its $p\rightarrow q$ norm as
\begin{align}
    \lVert V \rVert _{p\rightarrow q}= \underset{x\in V \setminus \{0\}}{\max}\frac{\lVert x \rVert_q}{\lVert x \rVert_p}.
\end{align}
Then, for all $q\geq 1$,
\begin{align}
   \lVert V \rVert _{\infty\rightarrow q} \geq \sqrt[q]{k}.
\end{align}
\end{lemma}

\begin{proof}
Let $Q$ be the hypercube $[-1,1]^n$. Then $V\cap Q$ is a bounded nonempty polytope, thus it must have a vertex. Denote that vertex as $x =(x_1,\ldots ,x_n)$. Without loss of generality, we may assume that there is some integer $0<r<n$ such that $x_i=1$ for all $ 1\leq i \leq r$ and $|x_i|<1$ for all $r+1\leq i \leq n$. This is since $Q$ is symmetric, thus we can always rearrange the coordinates of $V$ without affecting its norms.
We claim that $r\geq k$. Assume towards contradiction that it is not the case. Then there must be some nonzero vector $v \in V$ with $v_i=0$ for all $1\leq i \leq r$, such that for a small enough constant $\veps$, both $x+\veps v$  and $x-\veps v$  would be in $V\cap Q$, contradicting that $x$ is a vertex. Thus $\lVert x \rVert _{\infty}=1$ and $\lVert x \rVert _q \geq \sqrt[q]{r} \geq \sqrt[q]{k}$. 
\end{proof}
\textit{Proof of Lemma~\ref{lem:sameP}, continued}: Lemma~\ref{lem:ptoq} implies that there must be a vector $x\in V$ with $\lVert x \rVert_\infty=1/n$ and $\lVert x \rVert_1\geq k/n$.
For this $x$, we can define $q=u-x$. Recalling that $V$ is orthogonal to $[1 \ldots 1]$, we have $\sum q_i=\sum u_i-\sum x_i=1$, thus $q$ is a valid probability distribution and $\dtv{(u,q)}=\frac{1}{2}\lVert x \rVert _1\geq k/2n$. Since $x\in \ker(A)$, we have that the transition probability $q$ induces on the machine is 
\begin{align}
    p(q)=Aq=A(u-x)=A u=p(u).
\end{align}
Thus, for any finite state machine with $S=((1-2\veps)n-1)^{1/2}$ states, there exists a distribution $q$ with $\dtv{(u,q)}\geq \veps$ such that $\Prob(u)=\Prob(q)$. 
\end{proof}
\begin{proof}\textbf{of Theorem~\ref{thm:lower_bound}, irreducible case} :
Recall that the stationary row vector satisfies $\pi_u=\pi_u\Prob(u)$. This can be written as $\pi_u=Mp(u)=MA u$, where $M\in [0,1]^{S\times S^2}$ is of the form
\begin{align}
M=\begin{pmatrix}
\pi_u & 0 & \ldots &0\\
0 & \pi_u & \ldots & 0\\
\vdots & \vdots & \vdots & \vdots\\
0 & 0 & \ldots & \pi_u
\end{pmatrix}    
\end{align}
Denote $B=MA \in [0,1]^{S \times n}$. Similarly to Lemma~\ref{lem:sameP}, we have that if $S<n$, the matrix $B$ must have a nontrivial kernel of rank at least $n-S$. Letting $V_1$ be the subspace spanned by $\ker(B)\cap \ker([1 \ldots 1])$, we have that $\dim(V_1)\geq n-1-S \triangleq k$. Appealing to Lemma~\ref{lem:ptoq}, there is some valid probability distribution $q=u-x$ with $\dtv{(u,q)}\geq k/2n$ and $x\in \ker(B)$. We thus have that  $Bq=Bu=\pi_u$. But $Bq=\pi_u$ are simply the stationary equations for the chain under $q$, thus $\pi_q=\pi_u$.  
This implies that for any finite state machine with $S=(1-2\veps)n-1$ states, there exists a distribution $q$ with $\dtv{(u,q)}\geq \veps$ such that $\pi_u=\pi_q$.  
\end{proof}
\subsection{Proof of the \texorpdfstring{$\Omega(n)$}{Lg} bound for non-ergodic machines}
We now prove the case of  non-ergodic machines, corresponding to reducible Markov chains. Specifically, we show that any $(S,\delta)$ non-ergodic uniformity tester can be reduced to an $(S,\delta +\gamma)$ ergodic uniformity tester, where $\gamma$ can be made arbitrarily small. We start by considering a simple non-ergodic machine with only two absorbing states. 
\begin{lemma}\label{lem:make_ergodic}
Consider a Markov chain with initial state $s$ and two absorbing states $0$ and $1$, such that the absorption probabilities are $p_{\infty}^j=[\Pr(s\rightarrow 0|\m{H}_j), \Pr(s\rightarrow 1|\m{H}_j)]$ under hypothesis $\m{H}_j$. Further let $\mix$ be the maximal expected hitting time of state $i$ under $\m{H}_j$, given that we absorb in $i$, where the maximum is over all $i,j\in\{0,1\}$. Let $M>0$ be some positive constant. Then for any $j\in \{0,1\}$, the ergodic chain that results by connecting each absorbing state to $s$ with probability $\delta= \frac{1}{M\mix}$ has stationary probabilities $\pi_0^j,\pi_1^j$ such that
\begin{align}
  |\pi_i^j-p_{\infty}^j(i)|\leq  \frac{1}{M+1}\cdot \underset{i\in \{0,1\}}{\max}p_{\infty}^j(i) \leq \frac{1}{M+1} .
\end{align}
\end{lemma}
\begin{proof}
Consider the ergodic Markov chain obtained by connecting each of the absorbing states $0$ and $1$ to $s$ with probability $\delta= \frac{1}{M\mix}$. The transition rules of the transient states are left unchanged.
Assume w.l.o.g that the correct hypothesis is $\m{H}_0$, and for simplicity let $p_{\infty}=p_{\infty}^0,\pi_i=\pi_i^0 $. Appealing to~\cite{kac1947notion} [Theorem $2$ - Kac's return time], the stationary distribution of state $i\in \{0,1\}$ is $\pi_i=\frac{1}{\E[\tau_i]}$, where $\E[\tau_i]$ is the expected first return time to state $i$. The expected return time to state $0$ (under $\m{H}_0$) can be upper bounded by the following argument: with probability $1-\delta$ we return in one step, otherwise we play the following game. We flip a $\Bern(1-p_{\infty}(0))$ coin that determines whether we end in  state $0$ or not. If the coin value is $1$, we travel the chain at most $\mix$ time units on average and then another $\delta^{-1}$ time units on average 'stuck' in state $1$ before returning to state $s$ and flipping the coin again. If the coin value is $0$, we arrive at the desired state and terminate the game after at most $\mix$ time units on average. This is since the limiting distribution $p_{\infty}(0)$ is simply the probability of arriving at state $0$ before arriving at state $1$. 

Formally, let $W$ denote the number of rounds in the above game, i.e., the number of times we return to $s$ before hitting state $0$. Thus $\Pr(W=m)=(1-p_{\infty}(0))^{m-1}p_{\infty}(0)$, and furthermore, $\E[\tau_0|W=m]\leq \mix+m\mix+m\delta^{-1}$. We then have
\begin{align}
   \E[\tau_0]&=(1-\delta)+\delta \cdot \sum_{m=0}^\infty \Pr(W=m) \E[\tau_0|W=m]\\&= (1-\delta)+\delta \cdot p_{\infty}(0) \sum_{m=0}^{\infty}(1-p_{\infty}(0))^m\left(\mix+m\mix+m\delta^{-1}\right)\\&=(1-\delta)+\delta \cdot p_{\infty}(0)\left( \frac{\mix}{p_{\infty}(0)} +\left(\mix+\delta^{-1}\right)\sum_{m=0}^{\infty}m(1-p_{\infty}(0))^m\right)\\&= (1-\delta)+\delta \mix +\left(\delta \mix+1\right)\cdot \frac{1-p_{\infty}(0)}{p_{\infty}(0)}\\&\leq 1+\frac{\delta \mix+1-p_{\infty}(0)}{p_{\infty}(0)}=\frac{1+\delta \mix}{p_{\infty}(0)}.
\end{align}
Substituting $\delta= \frac{1}{M\mix}$ we have $p_{\infty}(0)-\pi_0\leq p_{\infty}(0) \cdot \frac{1}{M+1}$, which implies $\pi_1-p_{\infty}(1)\leq p_{\infty}(0) \cdot \frac{1}{M+1}$. Applying the above argument to $\E[\tau_1]$ concludes the proof.
\end{proof}
\begin{proof}\textbf{of Theorem~\ref{thm:lower_bound}, reducible case}:
Consider a Markov chain comprised of $K$ recurrent classes $\m{R}_1,\ldots,\m{R}_{K}$, and a set $\m{T}$ of transient states with initial state $s$. Assume w.l.o.g that $s\in\mathcal{T}$, as otherwise the machine is in fact irreducible.  Let $\mix$ be the maximum of the sum of the expected hitting time of class $\m{R}_k$ and the mixing time of the class under $\m{H}_j$, given that we absorb in $\m{R}_k$, where the maximum is over all $k\in [K],j\in \{0,1\}$. Motivated by the construction of Lemma~\ref{lem:make_ergodic} we connect all recurrent states (that is, all states not in $\m{T}$) to $s$ with probability $\delta=\frac{1}{M\mix}$, and with probability $1-\delta$ leave their transition rules unchanged, thus generating an ergodic Markov chain. Assume w.l.o.g that the correct hypothesis is $\m{H}_0$, and to simplify notations let $\Pe=\Pe(f,d|\m{H}_0)$ and $\Pr(s\rightarrow \m{R}_k)=\Pr(s\rightarrow \m{R}_k|\m{H}_0)$ for all $k\in [K]$. For each $k$, consider a chain with $|\mathcal{T}|+2$ states, obtained from the original chain by merging the states in $\m{R}_k$ and $\{\m{R}_j\}_{j\neq k}$ into two respectively absorbing states. By applying Lemma~\ref{lem:make_ergodic}, the stationary probability of the class $\m{R}_k$ for all $k\in [K]$, has
\begin{align}
  |\pi(\m{R}_k)-\Pr(s\rightarrow \m{R}_k)|\leq \frac{1}{M+1},  \label{eq:lem_bound}
\end{align}
where $\pi(\m{R}_k)$ is the stationary probability of the class $\m{R}_k$. We now bound the stationary probability of state $i \in \m{R}_k$ given class $\m{R}_k$, for all $k\in [K]$. To that end, let $\tilde{\pi}_i= \frac{\pi_i}{\pi(\m{R}_k)}$ for $i \in \m{R}_k$, and let $\pi_i^{\infty}$ be the limiting distribution of state $i$ in the original chain, given that we absorb in $\m{R}_k$. We show the following lemma, whose proof is relegated to the appendix.
\begin{lemma}\label{lem:new_pi}
For any $k\in [K]$ and $i \in \m{R}_k$,
\begin{align}
 \tilde{\pi}_i\geq \pi_i^{\infty}-\frac{|\m{R}_k|}{M}\left(1+\frac{1}{\pi(\m{R}_k)}\right).
\end{align}
\end{lemma}
Let $\Pr(\m{C})= \sum _{k=1}^K \Pr(s\rightarrow \m{R}_k)\Pr(\m{C}|\m{R}_k)$ be the asymptotic probability of success in the original chain, and let $\Pr(\m{C}^{\text{erg}})$ be the asymptotic probability of success in the new ergodic chain. Then
\begin{align}
    \Pr(\m{C}^{\text{erg}})&\geq \sum _{k=1}^K \pi(\m{R}_k)\Pr(\m{C}^{\text{erg}}|\m{R}_k)\\&\geq \sum _{k=1}^K\Pr(s\rightarrow \m{R}_k)\Pr(\m{C}^{\text{erg}}|\m{R}_k)-\sum _{k=1}^K\frac{1}{M+1}\label{eq:class_bound}\\&\geq \Pr(\m{C})-\sum _{k=1}^K\frac{|\m{R}_k|^2}{M}\left(1+\frac{1}{\pi(\m{R}_k)}\right)-\frac{K}{M+1}\label{eq:lemma9}
    \\&\geq\Pr(\m{C})-\frac{S^2\left(1+1/\pi(\m{R}_k)\right)+K}{M}\label{eq:sum_K},
\end{align}
where~\eqref{eq:class_bound} follows from~\eqref{eq:lem_bound},~\eqref{eq:lemma9} follows from $\Pr(\m{C}^{\text{erg}}|\m{R}_k)=\sum_{i:d(i)=\m{H}_0}\tilde{\pi}_i$ and by substituting Lemma~\ref{lem:new_pi}, and~\eqref{eq:sum_K} follows from $\sum _{k=1}^K|\m{R}_k|^2\leq S^2$. 
We conclude by recalling that Theorem~\ref{thm:lower_bound} for the irreducible case implies that any irreducible machine with probability of error smaller than $1/2$ must have $\Omega(n)$ states, therefore, by choosing $M$ to be sufficiently large, the bound 
\begin{align}
  \Pe ^{\text{erg}}\leq \Pe+\frac{S^2(1+1/\pi(\m{R}_k))+K}{M}  
\end{align}
implies that so should any reducible machine. 
\end{proof}

\acks{This work was supported by the ISF under Grants 1641/21 and 1495/18.}

\bibliography{uniformity}

\appendix
\section{Proof of Lemma~\ref{lem:biased}}
\begin{proof}
Assume toward contradiction that, for all $i$,
\begin{align}
    p_1^{(i,1)}&\leq \frac{1}{2}+\tilde{\veps} \Longrightarrow p_1\leq \frac{1+2\tilde{\veps}}{1-2\tilde{\veps}}\cdot p_i,\label{eq:p1_bound}\\
    p_i^{(i,1)}&\leq \frac{1}{2}+\tilde{\veps} \Longrightarrow p_i\leq\frac{1+2\tilde{\veps}}{1-2\tilde{\veps}}\cdot p_1.\label{eq:pi_bound}
\end{align}
~\eqref{eq:p1_bound} and~\eqref{eq:pi_bound} imply that
\begin{align}
 1&=\sum_{i=1}^n p_i\leq np_1 \cdot \frac{1+2\tilde{\veps}}{1-2\tilde{\veps}} ,\label{eq:upperb}\\
 1&=\sum_{i=1}^n p_i \geq np_1 \cdot \frac{1-2\tilde{\veps}}{1+2\tilde{\veps}}.\label{eq:lowerb}
\end{align}
Thus we have
\begin{align}
    \frac{1-2\tilde{\veps}}{1+2\tilde{\veps}}\cdot \frac{1}{n}\leq p_1\leq\frac{1+2\tilde{\veps}}{1-2\tilde{\veps}}\cdot \frac{1}{n},
\end{align}
and, as a result, for all $i>1$,
\begin{align}
    \left(\frac{1-2\tilde{\veps}}{1+2\tilde{\veps}}\right)^2\cdot \frac{1}{n}\leq p_i\leq \left(\frac{1+2\tilde{\veps}}{1-2\tilde{\veps}}\right)^2\cdot \frac{1}{n}.
\end{align}
Finally, we have
\begin{align}
    \dtv{(p,u)}=\sum_{i=1}^n \left|p_i-\frac{1}{n}\right|\leq \max \left\{\frac{8\tilde{\veps}}{(1-2\tilde{\veps})^2},\frac{8\tilde{\veps}}{(1+2\tilde{\veps})^2}\right\}=
    \frac{8\tilde{\veps}}{(1-2\tilde{\veps})^2}, 
\end{align}
and by substituting $\tilde{\veps}=\frac{\veps}{8+4\veps}$ we arrive at $\dtv{(p,u)}\leq \veps$, a contradiction. Thus, there must be some $i$ for which either $p_1^{(i,1)}$ or $p_i^{(i,1)}$ are greater than $\frac{1}{2}+\tilde{\veps}$.
\end{proof}

\section{Proof of Lemma~\ref{lem:new_pi}}
\begin{proof}
Assume w.l.o.g that $\m{R}_k=[m]$.
Let $\pi_i$ be the stationary distribution of state $i$ in the ergodic chain, and for all $i \in \m{R}_k$, let $\pi_i^{\infty}$ be the limiting distribution of state $i$ in the original chain, given that we absorb in $\m{R}_k$. Furthermore, let $p_{ij}$ be the transition probability from state $i$ to state $j$ in the original chain. The stationary equations for $\pi_j^{\infty}$ and $\pi_j$, for any $j\in \m{R}_k$, are
\begin{align}
  \pi_j^{\infty}&=\sum_{i\in \m{R}_k} \pi_i^{\infty}p_{ij}, \\ \pi_j&=\sum_{i\in \m{R}_k} \pi_i p_{ij}(1-\delta)+\sum _{i\in \m{T}}\pi_i p_{ij}.
\end{align}
Thus we have
\begin{align}
  \tilde{\pi}_j&=\sum_{i\in \m{R}_k} \tilde{\pi}_i p_{ij}(1-\delta)+\frac{1}{\pi (\m{R}_k)}\sum _{i\in \m{T}}\pi_i p_{ij}\\&=\sum_{i\in \m{R}_k} \tilde{\pi}_i p_{ij}+\underbrace{\frac{1}{\pi (\m{R}_k)}\sum _{i\in \m{T}}\pi_i p_{ij}-\delta\sum_{i\in \m{R}_k} \tilde{\pi}_i p_{ij}}_{e_j}.
\end{align}
Denote $\boldsymbol{\tilde{\pi}}=[\tilde{\pi}_1,\ldots,\tilde{\pi}_m],\boldsymbol{\pi}^{\infty}=[\pi_1^{\infty},\ldots,\pi_m^{\infty}],\boldsymbol{e}=[e_1,\ldots,e_m]$. Further Let $\Prob$ denote the transition matrix of the original chain restricted to $\m{R}_k$, and define $\boldsymbol{d} \triangleq \boldsymbol{\tilde{\pi}}-\boldsymbol{\pi}^{\infty}$. Note that since both $\boldsymbol{\pi}^{\infty}$ and $\boldsymbol{\tilde{\pi}}$ are valid probability distributions over $[m]$, we have that $\boldsymbol{d} \perp \boldsymbol{1}$. Writing the above equations in matrix form, we have
\begin{align}
  &\boldsymbol{d} =\boldsymbol{d}\Prob+\boldsymbol{e},\\&\boldsymbol{d}(\boldsymbol{I}-\Prob)=\boldsymbol{e}.
\end{align}
Let $\lambda_{\min}(\boldsymbol{I}-\Prob)$ be the minimal eigenvalue $\lambda\neq 0$ of $\boldsymbol{I}-\Prob$, and $\lambda_{\max}(\Prob)$ be the maximal eigenvalue $\lambda\neq 1$ of $\Prob$. Thus,
\begin{align}
    \lVert \boldsymbol{e} \rVert_2 = \lVert \boldsymbol{d}(\boldsymbol{I}-\Prob) \rVert_2 &\geq \lVert \boldsymbol{d}\rVert_2 \cdot \lambda_{\min}(\boldsymbol{I}-\Prob)\label{eq:lambda_bound}\\&=\lVert \boldsymbol{d}\rVert_2 (1-\lambda_{\max}(\Prob))\label{eq:spectral_gap},
\end{align}
where~\eqref{eq:lambda_bound} follows since for any matrix $\boldsymbol{A}$ operating on vector space $V$, its minimal eigenvalue has $\lambda_{\min}=\underset{\boldsymbol{v}\in V}{\min}\frac{\lVert \boldsymbol{A}\boldsymbol{v} \rVert_2}{\lVert \boldsymbol{v} \rVert_2}$, and~\eqref{eq:spectral_gap} follows since the eigenvalues of $\boldsymbol{I}-\Prob$ are $1-\lambda_i(\Prob)$ (recall that all the eigenvalues of a Markov transition matrix has $|\lambda_i|\leq 1$). The reason we dispense with $\lambda=0$, which is indeed an eigenvalue of $\boldsymbol{I}-\Prob$ with eigenvector $\boldsymbol{1}$, is that $\boldsymbol{d} \perp \boldsymbol{1}$, hence the largest contraction it can undergo is by $\lambda_{\min}(\boldsymbol{I}-\Prob)$. On the other hand,
\begin{align}
    \lVert \boldsymbol{e} \rVert_2 \leq m \cdot \lVert \boldsymbol{e} \rVert_{\infty}\leq m\cdot \delta\left(1+\frac{1}{\pi (\m{R}_k)}\right),\label{eq:norm_bound}
\end{align}
where the first inequality is a standard norm inequality, and the second is a result of the following lemma, which we introduce without the (simple) proof:
\begin{lemma}\label{lem:simple}
Let $\{X_n\}$ be a stationary process over some alphabet $\mathcal{S}$. Then for any disjoint partition $\mathcal{C}\cup \mathcal{C}'=\mathcal{S}$, it holds that
\begin{align}
\Pr(X_n\in\mathcal{C},X_{n+1}\in\mathcal{C}')=\Pr(X_n\in\mathcal{C}',X_{n+1}\in\mathcal{C}).
\end{align}
\end{lemma}
Appealing to the lemma with $\mathcal{C}=\m{T},\mathcal{C}'=\bigcup\limits_{k=1}^K \m{R}_k$, we have
\begin{align}
    \sum_{k=1}^K\sum_{j\in \m{R}_k}\sum_{i\in \m{T}}\pi_i p_{ij}=\sum_{k=1}^K\sum_{i\in \m{R}_k}\sum_{j\in \m{T}}\pi_i p_{ij}=\sum_{k=1}^K\sum_{i\in \m{R}_k}\pi_i \cdot \delta \leq \delta,
\end{align}
since the only transition from any $i\in \m{R}_k$ to $\m{T}$ is transition to $s$ with probability $\delta$. Thus we have, for any $j\in [K]$, 
\begin{align}
    |e_j|&\leq \frac{1}{\pi (\m{R}_k)}\sum _{i\in \m{T}}\pi_i p_{ij}+\delta\sum_{i\in \m{R}_k} \tilde{\pi}_i p_{ij}\\&\leq \frac{1}{\pi (\m{R}_k)}\sum_{k=1}^K\sum_{j\in \m{R}_k}\sum_{i\in \m{T}}\pi_i p_{ij}+\delta \\&\leq \frac{\delta}{\pi (\m{R}_k)}+\delta,
\end{align}
and~\eqref{eq:norm_bound} is established. Combining~\eqref{eq:spectral_gap} and~\eqref{eq:norm_bound}, and recalling that $\lVert \boldsymbol{d}\rVert_2\geq \lVert \boldsymbol{d}\rVert_{\infty}$, we conclude
\begin{align}
 \lVert \boldsymbol{d}\rVert_{\infty}&\leq \frac{m\cdot\delta}{1-\lambda_{\max}(\Prob)} \left(1+\frac{1}{\pi (\m{R}_k)}\right)\\&\leq m\cdot\delta (T_{\text{mix}}^k+1) \left(1+\frac{1}{\pi (\m{R}_k)}\right)\label{eq:mixing} \\&\leq \frac{m}{M}\left(1+\frac{1}{\pi (\m{R}_k)}\right),\label{eq:delta}
\end{align}
where $T_{\text{mix}}^k$ is the mixing time of the Markov chain with transition matrix $\Prob$,~\eqref{eq:mixing} follows from~\cite{levin2017markov}[Theorem 12.5], and~\eqref{eq:delta} follows from substituting $\delta=\frac{1}{M\mix}$. 
\end{proof}
\end{document}